\documentclass{article}
\usepackage[latin1]{inputenc}
\usepackage{amssymb, amsmath}
\usepackage[a4paper]{geometry}
\usepackage{pifont}
\usepackage{graphicx}
\usepackage{xcolor}
\usepackage{xtab} 
\usepackage{longtable} 
\usepackage{footnote}
\makesavenoteenv{tabular}
\usepackage{indentfirst}  
\usepackage{hyperref}
\hypersetup{
  colorlinks   = true,    
  urlcolor     = blue,    
  linkcolor    = blue,    
  citecolor    = red      
}

\newtheorem{definition}{Definition}
\newtheorem{proposition}{Proposition}
\newtheorem{remark}{Remark}

\newcommand{\aproof}{\hfill{\ding{111}}}
\newenvironment{proof}{\noindent{\sc{Proof:}}\newline}{\aproof}

\def\keywords{\vspace{.5em} 
{\textit{Keywords}:\,\relax%
}}

\author{L\'aszl\'o Csat\'o\thanks{Department of Operations Research and Actuarial Sciences, Corvinus University of Budapest, \newline
F\H{o}v\'{a}m t\'er 8., 1093 Budapest, Hungary \newline
e-mail: laszlo.csato@uni-corvinus.hu}}
\title{Ranking by pairwise comparisons for Swiss-system tournaments}
\date{\today}

\begin{document}

\maketitle

\begin{abstract}
Pairwise comparison matrices are widely used in Multicriteria Decision Making. This article applies incomplete pairwise comparison matrices in the area of sport tournaments, namely proposing alternative rankings for the 2010 Chess Olympiad Open tournament. It is shown that results are robust regarding scaling technique. In order to compare different rankings, a distance function is introduced with the aim of taking into account the subjective nature of human perception. Analysis of the weight vectors implies that methods based on pairwise comparisons have common roots. Visualization of the results is provided by Multidimensional Scaling on the basis of the defined distance. The proposed rankings give in some cases intuitively better outcome than currently used lexicographical orders.

\end{abstract}

\section{Introduction}
\label{Intro}

There are some sport tournaments where the ranking of the competitors is based on the results of games played against each other. In the world of sport there is no consensus in using a particular ranking method. Various evaluation methods have been applied to different events taking into account the specifics of the particular sport activity. Further complication emerges when participants are teams, and the final ranking should also reflect individual results. One characteristic example for that situation is a chess team championship.

In this article it is assumed that the final ranking of participants of a sport tournament is based on the outcome of the games that have been played against each other. The result of a game is given according to the rules of the particular sport, however, reasonable data transformation is also allowed. Ranking of the participants will be made by applying the pairwise comparison models of Multicriteria Decision Making methodology. It has two main advantages: the obvious interpretation of the games played against each other and the ability to address the problem of intransitivity (cyclical preferences regarding three alternatives: A is better than B, B is better than C, but C is better than A), a common feature of subjective evaluations by individuals, but a well-known phenomenon in objective sport results, too.

The 'alternatives' to be compared are the participants of the tournament, their game results will be incorporated into a pairwise comparison matrix, and the final result will follow the ranking of the priority vector derived from the matrix with an estimation method. It is worth mentioning that this approach can be applied not only for some real competitions (national soccer, basketball, hockey, etc. championships, chess tournaments), but for the ranking of teams or individuals having historic data of their rivalry (tennis or chess players, athletes).

Although the application of complete pairwise comparison matrices (i.e. everybody met all the other competitors at least once) can raise interesting questions to be analyzed (comparison of the 'official' rankings and those obtained by using the pairwise comparison method, for instance), the focus of this article is to use incomplete pairwise comparison matrices for the final ranking. This is the reason why the analyzed case is a Chess Olympiad: it is an ideal example of the potential application of incomplete pairwise comparison matrices.

Chess competitions are often organized as a Swiss-system tournament. All participants face each other for a determined number of games (rounds are often organized at the same time), while there is no knockout phase. It means that a loss in the first rounds does not make impossible for the contestant to win the championship, as well as all participants play the same number of matches. The final rank of the players (teams) is determined mainly by the application of different lexicographical orders, but it lacks a well established solution for determining the final ranking -- taking into account the performance of opponents of a team is a central issue as not all possible matchings materialized, which poses a serious challenge. It will be shown that an alternative way based on the results against each other could give a ranking which is in some sense intuitively better than the currently used methods. The more detailed analysis will take the results of the 39th Chess Olympiad Open tournament as a basis. The competition took place from September 20th to October 3rd, 2010 in Khanty-Mansiysk, Russia.

This choice was supported by the following arguments:
\begin{itemize}
	\item{It was an important sport event recently;}
	\item{Results were easy to collect and to adapt to the pairwise comparison method;}
	\item{Not only the winner, but the position of other participants were of interest;}
	\item{The reciprocity of the pairwise comparison matrix was a reasonable assumption;}
	\item{Participants were interested in the size of their win or lose;}
	\item{The official ranking method is debated.}
\end{itemize}
All of these issues will be thoroughly discussed.

The article is organized as follows. In Section \ref{Theory} the applied theory and some methods for the consistent approximation of incomplete pairwise comparison matrices will be described. Section \ref{Adapt} deals with the representation of the tournament results in a pairwise comparison format. Special features of the chess olympiad which make the necessary assumptions reasonable will also be presented here. Section \ref{Rankings} analyzes different rankings with a focus on comparing the official result with those of calculated from pairwise comparisons. For this purpose a distance measure will be introduced. Section \ref{Characteristics} highlights some interesting properties of the example. Section \ref{Visualization} uses Multidimensional Scaling (MDS) to draw the rankings on a two-dimensional map based on the distance defined in Section \ref{Rankings}. It reveals that the different rankings obtained from pairwise comparisons are close to each other. In Section \ref{Conclusion} the summary of the results and the outline of further research are given.

\section{Incomplete pairwise comparison matrices}
\label{Theory}

Pairwise comparisons have been widely used in decision making since Saaty published the AHP method \cite{Saaty1980}.

In a Multicriteria Decision Making problem the $n \times n$ real matrix $\mathbf{A}  = ( a_{ij} )$ is a pairwise comparison matrix if it is positive and reciprocal that is $a_{ij} > 0$ and $a_{ji} = 1/a_{ij}$ for all $i,j = 1,2, \dots ,n$. The reciprocity condition also means that $a_{ii} = 1$ for all $i = 1,2, \dots ,n$.

Matrix element $a_{ij}$ is the numerical answer to the question 'How many times is the $i$th alternative more important/better/favourable than the $j$th?' The final aim of the use of pairwise comparisons is to determine a weight vector $\mathbf{w} = (w_i)$ for the alternatives, where $w_i/w_j$ somehow approximates the pairwise comparison $a_{ij}$. The solution is obvious if matrix $\mathbf{A}$ is consistent, namely $a_{ik} = a_{ij}a_{jk}$ for all $i,j,k = 1,2, \dots ,n$, because there exists a positive $n$-dimensional vector $\mathbf{w}$ such that $a_{ij} = w_i/w_j$ for all $i,j = 1,2, \dots ,n$.

In the inconsistent case the real values of the decision maker can only be estimated. Saaty proposed the Eigenvector Method (EM) for this purpose, which is based on the Perron theorem \cite{Perron1907}, as a positive matrix has a dominant eigenvalue with multiplicity one and an associated strictly positive eigenvector.

Distance-minimization techniques, like Logarithmic Least Squares Method (LLSM) \cite{CrawfordWilliams1980, CrawfordWilliams1985, DeGraan1980} minimize the function $\sum_i \sum_j d(a_{ij}, w_i/w_j)$ where $d(a_{ij}, w_i/w_j)$ is the proper difference of $a_{ij}$ from its approximation $w_i/w_j$ (in case of LLSM $d(a_{ij}, w_i/w_j) = (\log a_{ij} - \log (w_i/w_j))^2$).

It may happen that a subset of pairwise comparisons are unknown due to the lack of available data, uncertain evaluations, or other problems. Incomplete pairwise comparison matrices were introduced in \cite{Harker1987}, for example (missing elements are denoted by $*$):

\[
 \mathbf{A}=
\left(
 \begin{array}{cccc}
 1		&   *       	&   a_{13}  	&   a_{14}  \\
 * 		&   1       	&   a_{23}  	&   *		\\
1/a_{13}&   1/a_{23}	&   1       	&   a_{34}  \\
1/a_{14}&   *       	&   1/a_{34}	&   1		\\
\end{array}
\right).
\]

In order to handle incomplete pairwise comparison matrices, introduce the variables \linebreak
$x_1, x_2, \dots ,x_d \in \mathbb{R}_+$ for the $d$ missing elements in the upper triangle of $\mathbf{A}$ supposing reciprocity (in all there are $2d$ unknown entries in the matrix). The new matrix is denoted by $ \mathbf{A(\mathbf{x})}$ as its elements are the functions of the variables:

\[
 \mathbf{A(x_1,x_2)}=
\left(
 \begin{array}{cccc}
1		&   x_1       	&   a_{13}  	&   a_{14}  \\
1/x_1	&   1       	&   a_{23}  	&   x_2		\\
1/a_{13}&   1/a_{23}	&   1       	&   a_{34}  \\
1/a_{14}&   1/x_2      	&   1/a_{34}	&   1		\\
\end{array}
\right).
\]

(In)complete pairwise comparison matrices can be represented by graphs \cite{Gass1998, Keri2011}. Let $\mathbf{A}$ be a pairwise comparison matrix of size $n \times n$. Then $G := (V,E)$, where $V = \{ 1,2, \dots ,n \}$ , vertices correspond to the alternatives, and $E = \{ e(i,j): a_{ij} \textrm{ is known and } i \neq j \}$, thus $E$ represents the structure of known elements.

A recent result in this field is the extension of EM and LLSM to the incomplete case \cite{BozokiFulopRonyai2010}. The EM solution is arising from the idea that inconsistency could be measured by the maximal eigenvalue, so it is a natural approach to choose the unknown elements for minimizing the dominant eigenvalue. Therefore the task is: $\lambda_{max} \left( \mathbf{A}(\mathbf{x}) \right) \rightarrow \min$, where $\mathbf{x} = (x_1, x_2, \dots ,x_d)^T \in \mathbb{R}_+^d$. As proved in \cite{BozokiFulopRonyai2010}, the solution is unique if and only if graph $G$ is connected, thus two alternatives could be compared via directly or indirectly through other alternatives. An algorithm is also proposed for the $\lambda_{max}$-optimal completion based on cyclic coordinates.

Another opportunity is the LLSM method. In this case only the known elements of $\mathbf{A}$ are examined: $\{ \sum_{e(i,j) \in E} \left[ \log a_{ij} - \log(w_i/w_j) \right]^2 \} \rightarrow \min$. The solution is unique again if and only if graph $G$ is connected, thus it depends strictly on the position of known elements, not on their exact values \cite{BozokiFulopRonyai2010}. Calculation of the optimal weights requires only the inverse of the upper-left $(n-1) \times (n-1)$ submatrix in the $\mathbf{L}$ Laplace-matrix of graph $G$ and some additional matrix multiplication. The problem can be solved quickly by any commonly used office softwares, like MS Excel.

\section{Application to the Chess Olympiad}
\label{Adapt}

An example for an incomplete pairwise comparison matrix is the result matrix of a chess olympiad described in the introduction. In the 39th Chess Olympiad 2010 Open tournament officially 149 teams participated. All of them played 11 matches, except for some teams due to their late arrivals.\footnote{Results are available in \cite{Results2010} in MS Excel or pdf-files. The data were organized into pairwise comparison matrices by the author; these are available on request.}

The matrix consists of 810 results from played matches, the ratio of known elements is small (7.3\%) as 11,026 elements is placed in the upper triangle of a $149 \times 149$ matrix. The diagonal elements of the matrix contains unities. There was no significant tendency in number of draws during the 11 rounds and it is the least frequent result. It implies that weaker teams played no 'simple' draws in the final rounds, when they still had no chances to reach a good position. This is important in chess, as draws could have arranged by the mutual agreement of players. However, that possibility is not totally excluded. The distribution of results is presented in Figure \ref{fig1}.

\begin{figure}[htbp]
\centering
\includegraphics[width = 10cm]{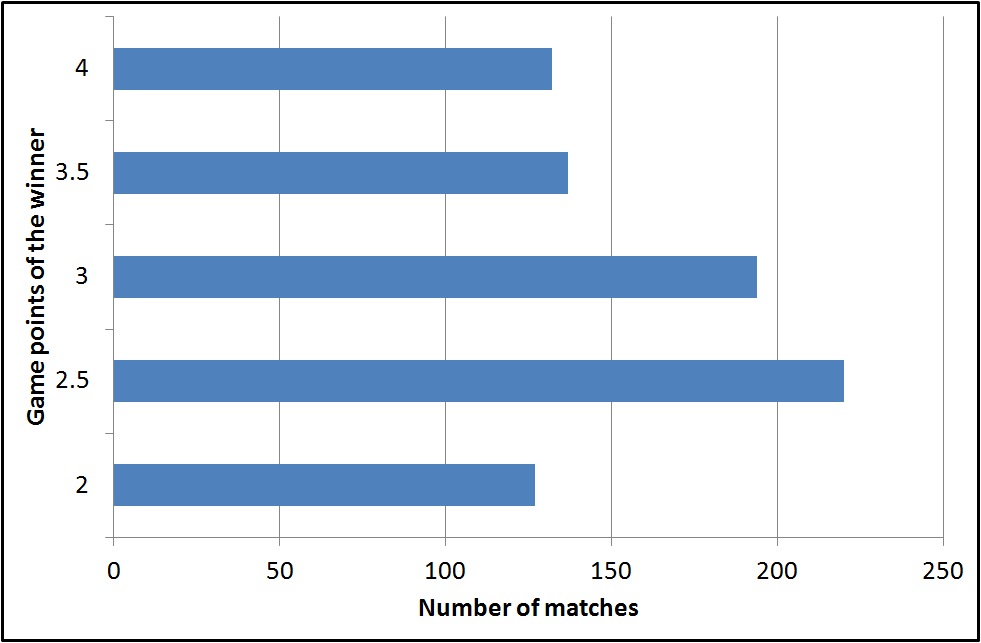} 
\caption{Distribution of game points of the winner (draws appear as 2)}
\label{fig1}
\end{figure}

In the tournament, all teams constituted by 5 players (4 and 1 reserve), a match between two of them includes 4 games with 3 different results (white win or loss, draw). In both teams 2 players have white, implying symmetricity as the chances of the match are a priori equal.\footnote{This is not true in individual chess tournaments, which is an important argument for examining the olympiad and similar team championships. Maybe the 5 Russian teams enjoyed some advantages due to the location, but it is not significant. In some sports it is false to suppose the reciprocity of the pairwise comparison matrix, like in soccer (home or away) or tennis (hard, carpet, clay or grass). The examination of pairwise comparison matrices without the reciprocity condition is out of the scope of the paper.} The winner of the game gets 1 game point, the loser 0, while draw means 0.5 game points for each player. Therefore the final result of the match for one team, the sum of players' game points, ranges from 0 to 4, by 0.5. If a team achieves in the match minimum 2.5 game points, it gets 2 match points. If it scores 2 game points (thus its opponent has likewise 2 game points), it gets 1 match point. If it secures 1.5 or fewer game points, it counts as 0 match points for this match. The sum of allocated match points is always 2.

The official ranking method is a lexicographical order determined by the application of 4 tie-breaking procedures in sequence, proceeding from TB1 to TB4 to the extent required.\footnote{\label{footnote5}The official ranking method is available in \cite{Handbook2010}. The position of teams that finish with the same number of match points shall be determined by application of the following tie-breaking procedures in sequence, proceeding from TB1 to TB4 to the extent required:
\begin{itemize}
	\item{TB1: Number of match points.}
	\item{TB2: Olympiad Sonneborn-Berger points. The number of each opponent's match points, excluding the opponent who scored the lowest number of matchpoints, multiplied by the number of game points achieved against this opponent.}
	\item{TB3: Number of game points. Their sum is always 4 in one match.}
	\item{TB4: Buchholz points. The sum of the match points of all the teams opponents, excluding the lowest one.}
\end{itemize}}
The primary criteria of the lexicographical order is TB1, the number of match points. However, as at most 22 match points could be scored in 11 matches, TB2 rules certainly should be applied, which means that teams are strongly interested in increasing their game points, players cannot be satisfied with a simple 2.5:1.5 win. Consequently, the size of wins reliably reflects the difference of teams' performance and it is justified to give higher weights to bigger wins. This is not the case in a lot of sport events.\footnote{It is true for tennis where players want to end the match as soon as possible, but not in football where goal difference usually does not count so much.}

It was presented that match results correspond to the main conception of ratios used in pairwise comparison matrices, and they could be transformed into values (ratios). Draws (2:2) are obviously converted to 1, but in the other 4 sort of results  4 different rules were applied (reciprocity ensures that it is enough to see the results from the viewpoint of the winner). The variants are presented in Table \ref{table2}.

\begin{table}[htbp]
\centering
\caption{Transformation of match results into pairwise comparison values (ratios)}
\label{table2}
\begin{tabular}{l|c|cccc}
    Game points & Number of matches & A variant & B variant & C variant & D variant \\ \hline
    0     & 132   &  1/5  &  1/8  &  1/3  &  1/5 \\
    0.5   & 137   &  1/4  &  1/6  &  2/5  &  1/4 \\
    1     & 194   &  1/3  &  1/4  &  1/2  &  2/7 \\
    1.5   & 220   &  1/2  &  1/2  &  2/3  &  1/3 \\
    2     & 127   & 1     & 1     & 1     & 1 \\
    2.5   & 220   & 2     & 2     & 1.5   & 3 \\
    3     & 194   & 3     & 4     & 2     & 3.5 \\
    3.5   & 137   & 4     & 6     & 2.5   & 4 \\
    4     & 132   & 5     & 8     & 3     & 5 \\
\end{tabular}
\end{table}

The size of wins weighs most in B variant, while in D ratios are not heavily effected by the number of game points. C variant contracts the scale, while A variant means the baseline scenario.

Now a method should be chosen to solve the problem, namely, to find a weight vector which approximates the incomplete pairwise comparison matrix relatively well. As it was mentioned in the previous section, EM and LLSM will be used. From a computational point of view, LLSM seems to be more favourable: the algorithm proposed by \cite{BozokiFulopRonyai2010} for the calculation of the right eigenvector is an iteration process in which a lot of univariable minimization problems should be solved. Therefore, in Section \ref{Rankings} the EM ranking will be discussed only for C variant. It was checked that the EM for other variants gives similar results to the LLSM, their difference is negligible compared to the official final ranking.

The appropriate result has another criteria, that graph $G$ corresponding to the pairwise comparison matrix should be connected. It depends only on the position of known elements, which is the same for the 4 variants. It cannot be decided in advance, but the solution of \cite{BozokiFulopRonyai2010} operates only for connected graphs. The pairing algorithm described in \cite{Handbook2010} suggests it will be connected and it was proved by the implementation of LLSM.

\section{Comparison of different rankings}
\label{Rankings}

The approximation of the incomplete pairwise comparison matrix is a 149-dimensional positive vector with normalized weights: $\mathbf{w} = (w_1, w_2, \dots , w_{149})^T \in \mathbb{R}^{149}_+$ where $\sum_{i=1}^{149} w_i = 1$. Elements are difficult to evaluate because of the large size of vectors, however, some features of the optimal weights will be discussed in Section \ref{Characteristics}. This section focuses on the relation of different rankings, gained from ordering of the elements of the priority vectors (full lists of countries sorted by positions according to the final ranking are presented in Table \ref{table10}). In the following, A-LLSM/B-LLSM/C-LLSM/D-LLSM corresponds to the ranking derived from the LLSM method for A/B/C/D variants. C-EM is the ranking obtained by the right eigenvector (EM) for C variant. Final ranking is the official final result of the Chess Olympiad. Start serves as a reference: it is the ranking of teams before the first round of the tournament based on the average FIDE (F\'ed\'eration Internationale des \'Echecs or World Chess Federation) rating of a team's players, which reflects their former performance.\footnote{See details in \cite{Handbook2010}.}

One of the most known index for comparing rankings is Spearman's rank correlation coefficient \cite{Spearman1904}. Let $X_i$ denote the rank of alternative $i$ according to ranking $X$ and $Y_i$ denote the rank of alternative $i$ according to ranking $Y$, then Spearman's rank correlation coefficient is:
\[
\rho = 1 - \frac{6 \sum_{i=1}^n (X_i-Y_i)^2}{n(n^2-1)}
\]
where $n$ is the maximal rank number. $\rho$ is the element of $\left[ -1,1 \right]$. These limits are reached when the two rankings are the same ($+1$) or entirely opposite ($-1$). $\rho = 0$ signs that there is no relation between the two rankings. $\rho^2$ can be interpreted as the fraction of variance shared between the two rankings. In this case there are no ties.\footnote{The official Start and Final rankings exclude ties by definition. Vectors approximating pairwise comparison matrices have no equal coordinates.}

Rank correlations are collected into a symmetric matrix, where the element of the $i$th row and the $j$th column is the Spearman's rank correlation coefficient between the corresponding rankings (see Table \ref{table3}).

\begin{table}[htbp]
\centering
\caption{Spearman's rank correlation coefficients between rankings}
\label{table3}
\begin{tabular}{l|ccccccc}
          & Start & Final & A-LLSM & B-LLSM & C-LLSM & D-LLSM & C-EM \\ \hline
    Start & 1     & 0.9353 & 0.9683 & 0.9684 & 0.9686 & 0.9654 & 0.9680 \\
    Final & 0.9353 & 1     & 0.9688 & 0.9686 & 0.9689 & 0.9699 & 0.9690 \\
    A-LLSM & 0.9683 & 0.9688 & 1     & 0.9997 & 0.9998 & 0.9987 & 0.9998 \\
    B-LLSM & 0.9684 & 0.9686 & 0.9997 & 1     & 0.9998 & 0.9978 & 0.9998 \\
    C-LLSM & 0.9686 & 0.9689 & 0.9998 & 0.9998 & 1     & 0.9983 & 0.9999 \\
    D-LLSM & 0.9654 & 0.9699 & 0.9987 & 0.9978 & 0.9983 & 1     & 0.9983 \\
    C-EM & 0.9680 & 0.9690 & 0.9998 & 0.9998 & 0.9999 & 0.9983 & 1 \\
\end{tabular}
\end{table}

The minimal $\rho$ among the 5 proposed rankings based on pairwise comparisons (in the $5 \times 5$ submatrix of Table \ref{table3}) is much bigger than 0.99: one of them explains the others at least in 99.5\% according to the interpretation of $\rho^2$. It seems to be a high value, however, there are some differences between them. For example, France is 6th in D-LLSM, while 10th in B-LLSM, or Jordan is 67th in D-LLSM, but 73rd in B-LLSM.

Rank correlations among the Final and the proposed rankings are above 0.96, the knowledge of the official result decrease the uncertainty regarding one of the latters at least by 93\% suggested by $\rho^2$. Nevertheless, great part of this is a simple fiction: an expert totally uninformed about this championship could still estimate the final result with more than 86\% certainty. It seems that Start  is nearer to LLSM rankings than to Final. It is not positive as the tournament will lose its curiosity if there are no surprises. On the other hand, the official method distorts in the direction of weaker teams confirmed by subsequent observations.

Spearman's rank correlation coefficient was used because it is a commonly accepted measure. In this case it has some disadvantages: $\rho$ is based on the square of rank differences, which means its value is determined mainly by the positions of weaker teams: the sum of rank differences between Start and Final rankings is 35,864, from which Pakistan (occupying the 123rd and 62nd positions, respectively) accounts for 3,721, more than 10\%. However, for the public the first 20 teams are more important than teams from the 80th to 99th positions.

The contextual factors of the situation should also be considered. Now it is more satisfactory to place the best team in the first position than to place the worst contestant last, known as ceiling effect \cite{Tarsitano2008}. There are a lot of measures of rank correlation which takes into account similar nonlinear effects. It would be a logical solution to use weights, but it still enlarges the differences in ranks due to squares and the chosen measure is intended to be used for another purpose, to visualize the rankings, which implies a rank correlation index is not appropriate.

In order to compare the above defined rankings, these problems should be addressed with a method satisfying the following soft properties:
\begin{itemize}
	\item{It regards the allocation of first places more important than the positions of weaker teams;}
	\item{It is symmetric and strictly monotonically increasing at every point, it increases if teams are positioned far from each other regardless to their exact positions;}
	\item{It possibly avoids the enlarging effect of square differences.}
\end{itemize}
Only rankings without ties will be discussed, i.e. supposing that all ties were broken by an arbitrary method like alphabetical order. It is a reasonable assumption is case of sport tournaments.

\begin{definition}
The $\tau$ measure between rankings $X$ and $Y$ is
\[
\tau(X,Y) = \sqrt{\sum_{i=1}^n \left( \log \frac{X_i}{Y_i} \right)^2}.
\]
Due to the properties of logarithm, it is the log-Euclidean metrics of the two rankings \cite{Arsigny2006}.
\end{definition}

\begin{proposition}
$\tau$ measure is a distance.
\end{proposition}

\begin{proof}
Denote $\epsilon (\mathbf{x},\mathbf{y}) = \sum_{i=1}^n (x_i - y_i)^2$ the standard Euclidean distance of vectors $\mathbf{x},\mathbf{y} \in \mathbb{R}^n$. Now $\tau(X,Y) = \epsilon(\log(X),\log(Y))$ where $\log(X)$ is a vector in $\mathbb{R}^n$ and $(\log(X))_i = \log X_i$. It means $\tau$ is a distance, too.
\end{proof}

It is not a new result, for a more general discussion see \cite{Norris2005}. In the following, it will be referred to as $\tau$ distance.

\begin{proposition}
$\tau$ has a maximum depending on the length of the ranking. $\tau$ has a unique maximum.
\end{proposition}

\begin{proof}
Take two rankings $X$ and $Y$. Suppose the $i$th alternative is better than the $j$th in both rankings: $X_i < X_j$ and $Y_i < Y_j$. It will be shown that the value of $\tau$ can be higher. In this analysis the final square root in the formula is irrelevant due to its monotonically increasing property.

Change the positions of $i$ and $j$ in $Y$, denote the new ranking by $Z$. Examine $\tau(X,Y)$ and $\tau(X,Z)$. The positions of all $k \neq i,j$ alternatives is the same in $Y$ and $Z$, the sum of squares for all $k \neq i,j$ is the same. Denote $\log X_i$ by $a_i$, $\log X_j$ by $a_j$, $\log Y_i$ by $b_i$ and $\log Y_j$ by $b_j$, then some calculation shows ($a_i < a_j$ and $b_i < b_j$ from the assumption:
\begin{eqnarray*}
\tau(X,Z) - \tau(X,Y) & = & (a_i - \log Z_i)^2 + (a_j - \log Z_j)^2 - (a_i - b_i)^2 -(a_j - b_j)^2 \\
 & = & (a_i - b_j)^2 + (a_j - b_i)^2 - (a_i - b_i)^2 -(a_j - b_j)^2 \\
 & = & -2 a_i b_j - 2 a_j b_i + 2 a_i b_i + 2 a_j b_j \\
 & = & 2 b_j (a_j - a_i) + 2 b_i (a_i - a_j) \\
 & = & 2 (b_j-b_i) (a_j - a_i) \geq 0.
\end{eqnarray*}

It implies that $\tau(X,Y)$ can be higher if there exists two alternatives $i$ and $j$ that $X_i < X_j$ and $Y_i < Y_j$. It is true for every $Y$ except the opposite ranking $X^{-1}$. Suppose that $X_i = i$ for all $i$ as indexing of the alternatives is arbitrary. If there exists no $i$ and $j$ that $X_i < X_j$ and $Y_i < Y_j$, then $X_1 = 1 \Rightarrow Y_1 = n$ because of $i=1$ and $j=2,3, \dots ,n$. Similarly, $X_2 = 2 \Rightarrow Y_2 = n-1$ because of $i=1$ and $j=3,4, \dots ,n$. It leads to the final conclusion that $Y=X^{-1}$. As the number of position changes is limited by $n \choose 2$, the iteration ends, and argument of the maximum is the two opposite rankings.
\end{proof}

\begin{remark}
The value of $\tau$ depends on the base of logarithm, which corresponds to a multiplying factor. In the following, the natural logarithm will be used.
\end{remark}

\begin{remark}
The $\tau$ distance satisfies the required conditions:
\begin{itemize}
    \item{It differentiates stronger in first places: if the first and second teams change their positions, $\tau^2$ evaluates it by $2\log(2) \approx 1.3863$, if the 80th and 90th teams switch places, $\tau^2$ increases only by $2\log(\frac{9}{8}) \approx 0.2357$. It is quite significant difference, but not meaningless -- people tend to record the best teams, while they do not bother about teams with average performance.}
    \item{It is symmetric and increases if teams positioned far from each other regardless to their exact positions as it is a distance.}
    \item{The logarithmic transformation contracts the scale rather than enlarging it. After that, Euclidean distance enlarges somewhat the differences but the concavity of logarithm is dominant for large differences. Among Start, Final and the other 5 rankings most $\max \{ X_i/Y_i;Y_i/X_i \}$ ratios are near to 1, where the logarithm can be approximated with the identity function.\footnote{The maximum of these ratios is 3.875 for Bulgaria, which is 8th in Start, but 31th in Final and D-LLSM. Start and Final are the farthest rankings, but in this relation only 7 are above 2 (among them the one for Bulgaria), while other 5 equals to 2. In the $\left[ 1,2 \right]$ interval the linear approximation of the logarithmic function is acceptable.} It means that the enlarging effect of squares remains high. For example, nearly 26-26\% of the total $\tau^2$ between A-LLSM and C-LLSM is due to Armenia and China occupying the 5th and 6th positions in both rankings.}
\end{itemize}
\end{remark}

$\tau$ distances between rankings are recorded in Table \ref{table4}. The proposed rankings are more or less at the same distance from Start and Final. Start and Final rankings are also the farthest as Spearman's rank correlation coefficient show. The difference between A-LLSM, B-LLSM and C-LLSM is negligible, while D-LLSM is somewhat far (it rewards mainly the victories, not their size). The EM method for C variant is nearly the same as LLSM rankings. Some countries have again great influence on the numbers: the $\tau^2$ between B-LLSM and D-LLSM is 1.013, from which 0.2609 derives from France (10th and 6th, respectively). Notably, some countries performed much better or worse compared to the knowledge of their players reflected by Start.

\begin{table}[htbp]
\centering
\caption{$\tau$ distances between rankings}
\label{table4}
\begin{tabular}{l|ccccccc}
          & Start & Final & A-LLSM & B-LLSM & C-LLSM & D-LLSM & C-EM \\ \hline
    Start & 0     & 4.008 & 2.928 & 2.926 & 2.909 & 3.150 & 2.923 \\
    Final & 4.008 & 0     & 2.806 & 2.817 & 2.795 & 2.788 & 2.806 \\
    A-LLSM & 2.928 & 2.806 & 0     & 0.489 & 0.359 & 0.672 & 0.383 \\
    B-LLSM & 2.926 & 2.817 & 0.489 & 0     & 0.262 & 1.007 & 0.259 \\
    C-LLSM & 2.909 & 2.795 & 0.359 & 0.262 & 0     & 0.896 & 0.142 \\
    D-LLSM & 3.150 & 2.788 & 0.672 & 1.007 & 0.896 & 0     & 0.913 \\
    C-EM & 2.923 & 2.806 & 0.383 & 0.259 & 0.142 & 0.913 & 0 \\
\end{tabular}
\end{table}

\begin{remark}
$\tau$ has a maximal value of $\sqrt{\sum_{i=1}^n \left[ \log(i) -\log (n+1-i) \right]^2}$, here $\tau_{\max}^{149} \approx 21.12$. It makes possible to normalize it, however, we have rankings with the same length, so it seems to be unnecessary. The theoretical maximum reflects the proximity of different rankings with respect to the contracting effect of logarithm dominant for high $\max \{ X_i/Y_i;Y_i/X_i \}$ ratios.
\end{remark}

The relation of different rankings could be examined by other statistical tools, as well. The position of teams in Final and A-LLSM rankings are drawn on Figure \ref{fig2} with linear regression analysis (the coefficient of explanatory variable $x$ is equal to Spearman's correlation coefficient between Final and A-LLSM rankings). Due to the similarity of the proposed rankings, the substitution of A-LLSM with another rankings calculated from the incomplete pairwise comparison matrix results in a similar chart. There is a remarkable tendency: teams with lower TB1 (match points), but higher TB2 (Sonneborn-Berger points) tend to achieve better positions than teams with the opponent performance benefit from the official lexicographical order.\footnote{See  footnote \ref{footnote5}.} It derives from lack of continuity of the lexicographical order.\footnote{It could be proved by simple intuition: let $a$ an alternative with values $(a_1,a_2,a_3 \dots , a_n)$ according to the $n$ criteria and choose a sequence of alternatives, where $b_k$ has ratings $(a_1-1/k,a_2+1,a_3+1, \dots , a_n+1)$. Let the limit of the sequence $b$, that is an alternative with values $(a_1,a_2+1,a_3+1 \dots , a_n+1)$. Then $a$ is better than $b_k$ for all $k$ according to the lexicographical order (if it prefers higher values), but $b$ is better than $a$. It contradicts to continuity, as the set of alternatives worse than $a$ is not closed.} For example, Zambia is officially the 47th, but at most 89th in LLSM and EM rankings, while Georgia stands at the 30th position in Final and is at least the 16th according to the proposed solutions. Other examples are given in Table \ref{table10}.

\begin{figure}[htbp]
\centering
\includegraphics[width = 10cm]{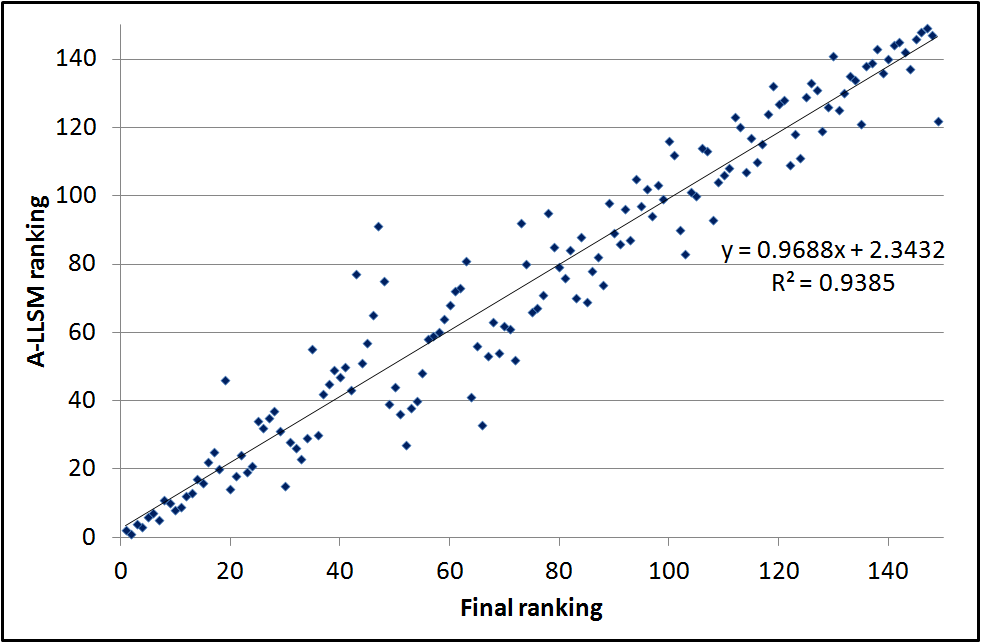} 
\caption{Relation of Final and A-LLSM rankings}
\label{fig2}
\end{figure}

Another important issue is the allocation of first places. In this respect, the results are robust as the first four positions are occupied by Ukraine, Russia 1, Hungary and Israel in Final and all of the proposed rankings.

\section{Some characteristics of the example}
\label{Characteristics}

Despite the similar rankings derived from LLSM vectors, the priority vectors are different for the 4 variants; Table \ref{table7} shows some details. The means of optimal weights are equal due to the normalization, while standard deviations confirm intuition: it is the largest in B variant which uses the widest scale and the smallest in 'narrow' C variant. For B and D variants the weight of the winner is similar, but its superiority (Max / Min ratio) to the last team is more than double for B. However, a factor of 250 or 500 refers to significant differences among teams.

\begin{table}[htbp]
\centering
\caption{Features of LLSM vectors}
\label{table7}
\begin{tabular}{l|cccc}
                        & A-LLSM & B-LLSM & C-LLSM & D-LLSM \\ \hline
    Maximum             & 0.0336 & 0.0422 & 0.0210 & 0.0417 \\
    Minimum             & 0.0002 & 0.0001 & 0.0007 & 0.0002 \\
    Max / Min           & 158.62 & 576.08 & 28.02 & 249.83 \\
    Mean                & 0.0067 & 0.0067 & 0.0067 & 0.0067 \\
    Standard deviation  & 0.0068 & 0.0082 & 0.0045 & 0.0078 \\
    Average win's ratio & 3.2650 & 4.5300 & 2.1325 & 3.7291 \\
    Power               & 4.2818 & 4.2074 & 4.4012 & 4.1946 \\
\end{tabular}
\end{table}

The ratio between the maximal and minimal weights has remarkable implications highlighted in the last 2 rows of Table \ref{table7}. Average win's ratio corresponds to the mean of wins in each coding:
\[
\sum_{i=2.5, 3, 3.5, 4} \frac{\text{Number of matches where winner's game points is } i * \text{Ratio corresponding to result } i}{\text{Number of matches without draws}}
\]
For example, it is $(220*2+194*4+137*6+132*8)/683$ for B variant. Finally, power equals to $\log(\text{Max / Min})/\log(\text{Average win's ratio})$, which reflects a kind of 'order of magnitude' in the tournament. For example, in A variant the factor between the first and the last team is 150-fold, while a standard victory in a match corresponds to a pairwise comparison ratio of 3.2. As $3.2^4 \approx 150$, a 'virtual' chain of 5 teams should exist, where each team defeated the following with this average difference. However, if one of the first and last teams played against each other, the ratio as a known element in the incomplete pairwise comparison matrix is at most 5 according to a 4:0 win in A variant. It implies participants should be put into 4-5 groups. Figure \ref{fig6} presents the weights derived from LLSM vectors.

\begin{figure}[htbp]
\centering
\includegraphics[width = 10cm]{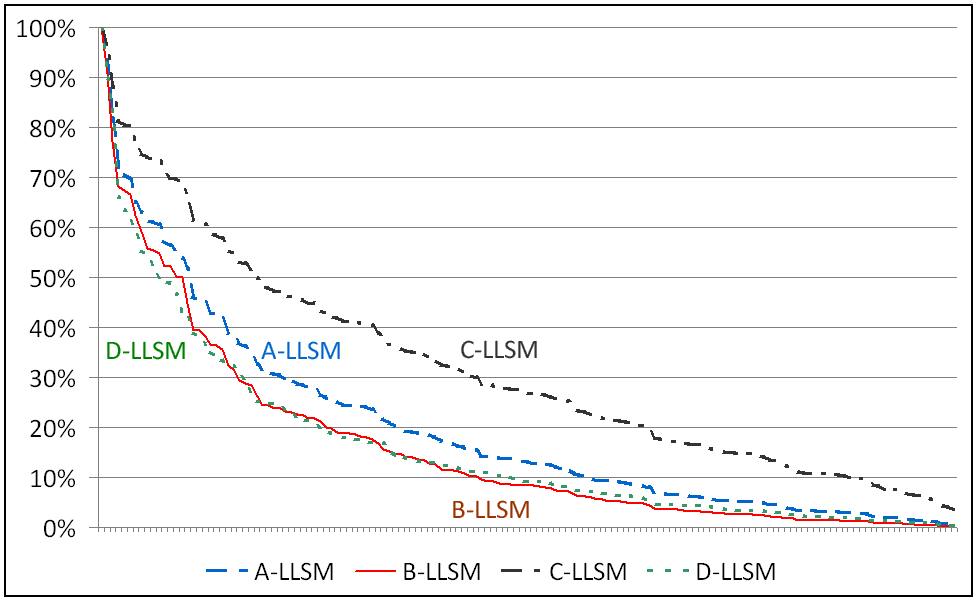} 
\caption{Weights of LLSM vectors relative to the first}
\label{fig6}
\end{figure}

Due to the pairing method, the classification of teams partly happens. The structure of known elements in the pairwise comparison matrix shows an impressive picture; Figure \ref{fig4} indicates them with a filled box covering approximately 7\% of the whole area, which is the ratio of known elements relative to all possible matchings. On the two charts matrix rows (teams) are ordered according to Final (a, left) and A-LLSM (b, right) rankings, respectively. Known elements representing matches knit around the diagonal, this effect is more stronger in the second case. It is affirmed by some calculations: the average difference between the rank of two teams played against each other is 28.70 in Final and 25.32 in A-LLSM, while the median is 22.5 and 19. Thus matches were taken place between contestants with a similar performance, the matching algorithm operates not randomly. Despite the fact that the pairing method was not examined, the above numbers suggest that the proposed pairing based on pairwise comparisons is better regarding the classification of teams -- if the number of matches is limited, they should be played between participants whose internal ranking is difficult to decide.

\begin{figure}[htbp]
\centering
\resizebox{\textwidth}{!}{
\includegraphics{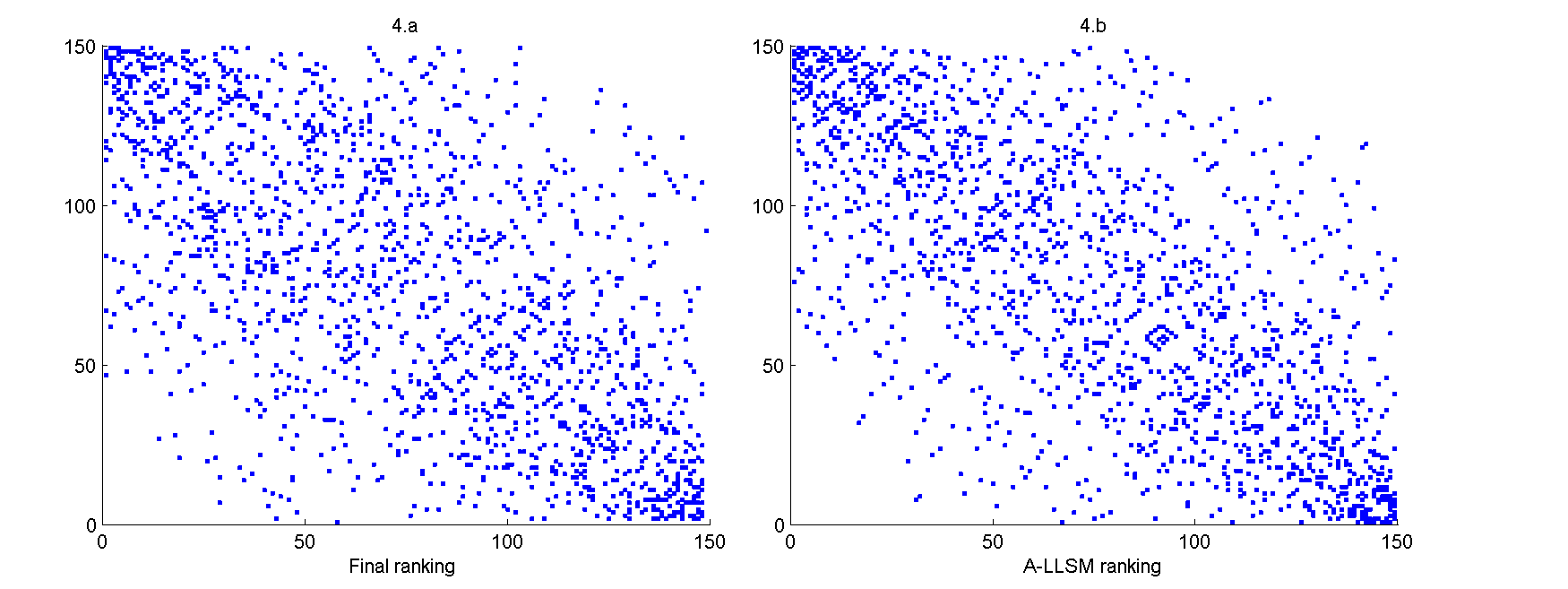} }
\caption{Map of known elements in the pairwise comparison matrix, teams sorted by positions according to Final (a, left) and A-LLSM (b, right) rankings}
\label{fig4}
\end{figure}

Amid these circumstances, it is favourable that the number of rounds in the tournament remains small compared to the number of participants, as it makes possible the appearance of significant ratios not permitted by the transformation of match results into pairwise comparison ratios. It raises the question whether all countries should participate in the same tournament which is not the case in almost all team sports like soccer or ice hockey.

\section{Visualization}
\label{Visualization}

Rankings based on pairwise comparison matrices have revealed some interesting facts about the performance of certain teams, nonetheless, the long lists do not give an impressive picture about their relation. The $\tau$ distance seems to be a good starting point, as it makes possible to plot $n$ rankings in a $n-1$ dimensional space. In addition, it is possible to decrease number of dimensions in order to give a better overview if it goes without significant loss of information. Multidimensional Scaling (MDS) is a statistical method in information visualization for exploring similarities or dissimilarities in data; a basic book in this field is \cite{KruskalWish1978}. A classical application of MDS is to draw cities on a map from the matrix consisting of their air distance.

For sufficiently small $n$, it requires only 1 or 2 dimensions to visualize all points representing different rankings. The distance matrix of the 7 rankings is still calculated in Section \ref{Rankings}, but the applied software (SPSS v18.0) requires at least 10 cases to be evaluated. Thus, regarding former observations, 3 further rankings (Sonneborn-Berger, Buchholz, Mix) were defined on the basis of different components of the original lexicographical order.\footnote{See footnote \ref{footnote5}.
\begin{itemize}
	\item{Sonneborn-Berger: a lexicographical order based on the Sonneborn-Berger points (TB2) of teams, then match points (TB1) and game points (TB3). It still gives a complete order.}
	\item{Buchholz: another lexicograhical order firstly based on the product of Buchholz points (TB4) and match points (TB1) divided by the number of matches. The idea is that the Buchholz points reflect only the force of opponents, so it should be modified. It is pointless to multiply it with game points (TB3) as it gives almost the Sonneborn-Berger points (TB2) and teams are interested in the increase of their match points (TB1). The subsequent tie-breaking rule is number of match points (TB1), then Sonneborn-Berger points (TB2).}
	\item{Mix: a composite index based on Sonneborn-Berger (TB2) and Buchholz points (TB4). In order to measure them in a similar scale, Buchholz points are multiplied by a correcting factor:
\[
F = \frac{3*\text{Number of wins}+2*\text{Number of draws}+1*\text{Number of losses}}{\text{Number of matches}}
\]
Then Sonneborn-Berger (with an average of 203) and modified Buchholz points (average 232) are added. The unique tie-break between Scotland and Faroe Islands is decided for the latter because of higher Sonneborn-Berger points.}
\end{itemize}

The 3 rankings partly moderate the oddities of the final ranking. For example, the officially 47th Zambia is 72nd, 75th and 79th in the Sonneborn-Berger, Buchholz and Mix rankings, while the 64th Germany is 37th, 60th and 45th, respectively.}

MDS maps the original distances in interval, ratio or ordinal scales; the most general interval scale were applied. Here $\delta$ discrepancies on the reduced-dimension map are related to the original $d$ distances by a linear function: $\delta = a + bd$, which means it is indifferent to multiplying the distances by a constant factor caused by choosing a different base of logarithm for $\tau$. The goodness of mapping (the information loss derives from dimension reduction) is measured by Kruskal's Stress and RSQ.

The pairwise $\tau$ distances 10 rankings could be mapped appropriately in two dimensions, but one dimension is not enough. The coordinates on the MDS map has no further meaning, distances on the map reflects the original distances as well as possible, while the mean of $x$ and $y$ coordinates is 0. Kruskal's Stress is 0.1431, a middle-strength relation. RSQ reaches 0.9468, approximately 94.5\% of variance of the scaled data can be accounted for by the MDS procedure. The reduced map is plotted in Figure \ref{figure3}.

\begin{figure}[htbp]
\centering
\includegraphics[width = 10cm]{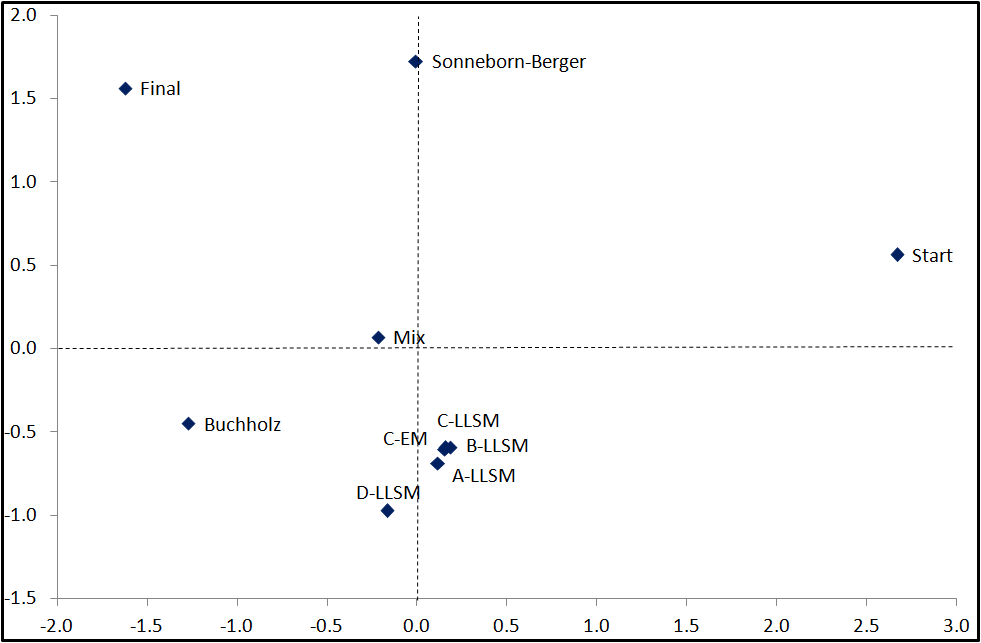} 
\caption{Rankings on a two-dimensional map by $\tau$ distances}
\label{figure3}
\end{figure}

Start and Final rankings are the farthest from each other, the 5 rankings derived from the incomplete pairwise comparison matrices are in a small cluster approximately at the same distance from Start and Final. The proposed rankings are almost the same, only D-LLSM differs from the others marginally (it gives a great weight to all wins, but hardly rewards its size). The 3 newly defined rankings (Sonneborn-Berger, Buchholz and Mix) are nearer to the Final than to the Start which is reassuring, while Buchholz and Mix are still nearer to the proposed rankings. It backs the former observation that methods based on pairwise comparison matrices overweighs TB4 (Buchholz points) to TB2 (Sonneborn-Berger points), while the official lexicographical order prefers teams with higher TB2 (Sonneborn-Berger points).\footnote{See footnote \ref{footnote5}.} It is not by chance, as TB4 (Buchholz points) rather reflects the performance of opponents, the key idea beyond pairwise comparison matrices.

\section{Conclusion}
\label{Conclusion}

This paper presents an alternative method to determine the final ranking of teams participated in the 39th Chess Olympiad 2010 Open tournament based on incomplete pairwise comparison matrices. In lack of former experience, 4 arbitrary scales were used to transform the match results into ratios. Two well-known techniques are used, the Eigenvector Method (EM) and the Logarithmic Least Squares Method (LLSM). The decision between them is far from trivial, while other methods deserve testing, like the one given by Fedrizzi and Giove in \cite{Fedrizzi2007}.

Results show that in some cases the incompleteness of the pairwise comparison matrix has favourable effects if it has a special structure, namely, if ordering on the priority vector derived from the matrix, known elements are located near the diagonal with high probability: some ratios given by the optimal completion can significantly exceed the original scale.

The chosen variant for transforming match results into ratios affects the absolute priorities significantly, but final rankings served by weights are relatively robust. It was confirmed by the application of a distance based on the asymmetry of committed mistakes in rankings, however, it was strongly influenced by some particular teams. The construction of better indices is a topic of future research.

Matches with the same result were represented by the same ratios. It means some simplification as a 4:0 win in the final round against a strong opponent seems to be more valuable than a similar result in the first rounds against underdogs. However, the official ranking does not take into account the opponent's strength in awarding victories, and a priori it is not reasonable to transform these outcomes differently. Nevertheless, further investigations are needed how this aspect could be incorporated into the pairwise comparison method.

Alternative proposals revealed some oddities of the current FIDE olympiad tie-breaking rules. As there is no commonly accepted ranking method in chess, other techniques are worth examining. Visualization implies that certain lexicographical orders, despite the lack of continuity, approximates the results of the proposed method well. While it is really impossible to find a perfect final ranking for a similarly complex Swiss-system tournament, it was justified that in some cases the use of pairwise comparison matrices give robust and intuitively better results than currently used ranking techniques. This approach can be extended to other sport championships.

\section*{Acknowledgements}

The author is very grateful to the anonymous reviewers for their valuable comments and suggestions. The research was supported by OTKA grant K-77420.

\newpage

\section*{Appendix}

\footnotesize
\begin{center}
\topcaption{Positions of teams according to different rankings}
\label{table10}
\xentrystretch{-0.13}
\tablefirsthead{\textbf{Team} & \textbf{Start} & \textbf{Final} & \textbf{A-LLSM} & \textbf{B-LLSM} & \textbf{C-LLSM} & \textbf{D-LLSM} \\ \hline}
\tablehead{\multicolumn{7}{c}{\tablename \  \thetable \ -- continued from previous page} \\
\textbf{Team} & \textbf{Start} & \textbf{Final} & \textbf{A-LLSM} & \textbf{B-LLSM} & \textbf{C-LLSM} & \textbf{D-LLSM} \\ \hline}
\tabletail{\hline \multicolumn{7}{r}{Continued on next page} \\ \hline}
\tablelasttail{\hline \hline}
\begin{xtabular}{lcccccc}
    Ukraine & 2     & 1     & 2     & 2     & 2     & 2 \\
    Russia 1 & 1     & 2     & 1     & 1     & 1     & 1 \\
    Israel & 11    & 3     & 4     & 4     & 4     & 4 \\
    Hungary & 5     & 4     & 3     & 3     & 3     & 3 \\
    China & 3     & 5     & 6     & 5     & 5     & 7 \\
    Russia 2 & 4     & 6     & 7     & 7     & 7     & 9 \\
    Armenia & 6     & 7     & 5     & 6     & 6     & 5 \\
    Spain & 16    & 8     & 11    & 9     & 10    & 10 \\
    United States of America & 9     & 9     & 10    & 11    & 11    & 8 \\
    France & 10    & 10    & 8     & 10    & 9     & 6 \\
    Poland & 15    & 11    & 9     & 8     & 8     & 11 \\
    Azerbaijan & 7     & 12    & 12    & 12    & 12    & 13 \\
    Russia 3 & 14    & 13    & 13    & 14    & 13    & 12 \\
    Belarus & 35    & 14    & 17    & 17    & 17    & 17 \\
    Netherlands & 13    & 15    & 16    & 16    & 16    & 14 \\
    Slovakia & 22    & 16    & 22    & 22    & 22    & 22 \\
    Brazil & 24    & 17    & 25    & 26    & 25    & 24 \\
    India & 19    & 18    & 20    & 21    & 20    & 20 \\
    Denmark & 44    & 19    & 46    & 47    & 47    & 42 \\
    Czech Republic & 17    & 20    & 14    & 13    & 14    & 15 \\
    Italy & 30    & 21    & 18    & 19    & 18    & 18 \\
    Greece & 25    & 22    & 24    & 23    & 24    & 25 \\
    Cuba  & 18    & 23    & 19    & 18    & 19    & 19 \\
    England & 12    & 24    & 21    & 20    & 21    & 21 \\
    Argentina & 26    & 25    & 34    & 34    & 34    & 33 \\
    Estonia & 48    & 26    & 32    & 31    & 32    & 32 \\
    Kazakhstan & 41    & 27    & 35    & 36    & 36    & 35 \\
    Moldova & 31    & 28    & 37    & 38    & 37    & 36 \\
    Iran  & 38    & 29    & 31    & 32    & 31    & 30 \\
    Georgia & 20    & 30    & 15    & 15    & 15    & 16 \\
    Bulgaria & 8     & 31    & 28    & 28    & 28    & 31 \\
    Croatia & 28    & 32    & 26    & 25    & 26    & 26 \\
    Serbia & 21    & 33    & 23    & 24    & 23    & 23 \\
    Sweden & 34    & 34    & 29    & 30    & 30    & 28 \\
    Lithuania & 39    & 35    & 55    & 55    & 55    & 58 \\
    Slovenia & 29    & 36    & 30    & 29    & 29    & 29 \\
    Canada & 53    & 37    & 42    & 42    & 42    & 41 \\
    Austria & 45    & 38    & 45    & 43    & 46    & 48 \\
    Russia 4 & 52    & 39    & 49    & 49    & 49    & 51 \\
    Iceland & 54    & 40    & 47    & 46    & 45    & 44 \\
    Egypt & 40    & 41    & 50    & 50    & 50    & 50 \\
    Montenegro & 56    & 42    & 43    & 44    & 43    & 43 \\
    Qatar & 55    & 43    & 77    & 77    & 76    & 79 \\
    Peru  & 46    & 44    & 51    & 52    & 51    & 49 \\
    Turkey & 50    & 45    & 57    & 58    & 58    & 54 \\
    Uruguay & 74    & 46    & 65    & 66    & 65    & 65 \\
    Zambia & 121   & 47    & 91    & 89    & 91    & 92 \\
    ICSC  & 72    & 48    & 75    & 76    & 75    & 73 \\
    Uzbekistan & 33    & 49    & 39    & 39    & 39    & 40 \\
    Philippines & 37    & 50    & 44    & 45    & 44    & 46 \\
    Norway & 23    & 51    & 36    & 35    & 35    & 37 \\
    Vietnam & 27    & 52    & 27    & 27    & 27    & 27 \\
    Chile & 51    & 53    & 38    & 37    & 38    & 38 \\
    Colombia & 57    & 54    & 40    & 40    & 40    & 39 \\
    Australia & 49    & 55    & 48    & 48    & 48    & 47 \\
    Former YUG Rep of Macedonia & 43    & 56    & 58    & 57    & 57    & 61 \\
    Albania & 68    & 57    & 59    & 59    & 59    & 57 \\
    Singapore & 73    & 58    & 60    & 60    & 60    & 55 \\
    Finland & 60    & 59    & 64    & 64    & 64    & 63 \\
    Belgium & 71    & 60    & 68    & 69    & 69    & 70 \\
    United Arab Emirates & 88    & 61    & 72    & 70    & 71    & 72 \\
    Pakistan & 123   & 62    & 73    & 74    & 74    & 71 \\
    IPCA  & 70    & 63    & 81    & 81    & 81    & 80 \\
    Germany & 42    & 64    & 41    & 41    & 41    & 45 \\
    Switzerland & 47    & 65    & 56    & 54    & 54    & 62 \\
    Bosnia \& Herzegovina & 32    & 66    & 33    & 33    & 33    & 34 \\
    Indonesia & 67    & 67    & 53    & 53    & 53    & 56 \\
    Kyrgyzstan & 77    & 68    & 63    & 61    & 63    & 64 \\
    Latvia & 58    & 69    & 54    & 56    & 56    & 53 \\
    Russia 5 & 61    & 70    & 62    & 63    & 62    & 60 \\
    Mongolia & 66    & 71    & 61    & 62    & 61    & 59 \\
    Mexico & 36    & 72    & 52    & 51    & 52    & 52 \\
    Bangladesh & 82    & 73    & 92    & 92    & 92    & 89 \\
    South Africa & 81    & 74    & 80    & 80    & 80    & 78 \\
    Portugal & 59    & 75    & 66    & 65    & 66    & 66 \\
    Turkmenistan & 69    & 76    & 67    & 68    & 68    & 69 \\
    Jordan & 79    & 77    & 71    & 73    & 72    & 67 \\
    Libya & 105   & 78    & 95    & 95    & 94    & 95 \\
    Paraguay & 84    & 79    & 85    & 85    & 85    & 84 \\
    Faroe Islands & 83    & 80    & 79    & 79    & 79    & 75 \\
    Venezuela & 64    & 81    & 76    & 78    & 78    & 74 \\
    Costa Rica & 80    & 82    & 84    & 84    & 84    & 83 \\
    Scotland & 63    & 83    & 70    & 72    & 70    & 68 \\
    Yemen & 85    & 84    & 88    & 90    & 89    & 85 \\
    Ecuador & 65    & 85    & 69    & 67    & 67    & 76 \\
    Tajikistan & 62    & 86    & 78    & 75    & 77    & 81 \\
    Andorra & 89    & 87    & 82    & 82    & 82    & 82 \\
    Ireland & 75    & 88    & 74    & 71    & 73    & 77 \\
    Algeria & 91    & 89    & 98    & 97    & 98    & 99 \\
    Dominican Republic & 87    & 90    & 89    & 87    & 87    & 93 \\
    New Zealand & 92    & 91    & 86    & 86    & 86    & 87 \\
    Malaysia & 86    & 92    & 96    & 96    & 96    & 96 \\
    Thailand & 94    & 93    & 87    & 88    & 88    & 91 \\
    Panama & 107   & 94    & 105   & 105   & 104   & 101 \\
    Barbados & 96    & 95    & 97    & 98    & 97    & 97 \\
    Japan & 101   & 96    & 102   & 102   & 102   & 103 \\
    Luxembourg & 90    & 97    & 94    & 94    & 95    & 94 \\
    Cyprus & 109   & 98    & 103   & 103   & 103   & 100 \\
    Guatemala & 103   & 99    & 99    & 100   & 99    & 98 \\
    Malta & 106   & 100   & 116   & 116   & 116   & 114 \\
    Nigeria & 139   & 101   & 112   & 114   & 113   & 109 \\
    IBCA  & 78    & 102   & 90    & 91    & 90    & 88 \\
    Iraq  & 76    & 103   & 83    & 83    & 83    & 86 \\
    Sri Lanka & 115   & 104   & 101   & 101   & 101   & 104 \\
    Jamaica & 97    & 105   & 100   & 99    & 100   & 107 \\
    Uganda & 127   & 106   & 114   & 115   & 114   & 113 \\
    Nepal & 114   & 107   & 113   & 110   & 112   & 116 \\
    Puerto Rico & 100   & 108   & 93    & 93    & 93    & 90 \\
    Lebanon & 99    & 109   & 104   & 104   & 105   & 102 \\
    Monaco & 95    & 110   & 106   & 106   & 106   & 106 \\
    Honduras & 125   & 111   & 108   & 108   & 108   & 112 \\
    Palestine & 134   & 112   & 123   & 123   & 122   & 124 \\
    Korea & 116   & 113   & 120   & 120   & 120   & 120 \\
    Bolivia & 104   & 114   & 107   & 107   & 107   & 105 \\
    Trinidad \& Tobago & 108   & 115   & 117   & 117   & 117   & 115 \\
    Botswana & 102   & 116   & 110   & 112   & 110   & 110 \\
    Bahrain & 112   & 117   & 115   & 113   & 115   & 117 \\
    Mauritius & 110   & 118   & 124   & 122   & 123   & 126 \\
    Chinese Taipei & 131   & 119   & 132   & 132   & 132   & 132 \\
    Kenya & 133   & 120   & 127   & 126   & 126   & 127 \\
    Aruba & 120   & 121   & 128   & 128   & 128   & 130 \\
    Wales & 93    & 122   & 109   & 111   & 111   & 108 \\
    Jersey & 113   & 123   & 118   & 118   & 118   & 111 \\
    Angola & 98    & 124   & 111   & 109   & 109   & 118 \\
    Mali  & 145   & 125   & 129   & 129   & 129   & 128 \\
    Namibia & 129   & 126   & 133   & 133   & 133   & 134 \\
    Malawi & 141   & 127   & 131   & 131   & 131   & 131 \\
    Ethiopia & 132   & 128   & 119   & 119   & 119   & 119 \\
    Hongkong & 124   & 129   & 126   & 127   & 127   & 123 \\
    Guernsey & 128   & 130   & 141   & 142   & 142   & 141 \\
    Mauritania & 146   & 131   & 125   & 124   & 124   & 125 \\
    Surinam & 111   & 132   & 130   & 130   & 130   & 129 \\
    Macau & 122   & 133   & 135   & 135   & 135   & 136 \\
    Mozambique & 130   & 134   & 134   & 134   & 134   & 133 \\
    Madagascar & 140   & 135   & 121   & 121   & 121   & 122 \\
    Netherlands Antilles & 118   & 136   & 138   & 139   & 139   & 138 \\
    Cameroon & 144   & 137   & 139   & 137   & 137   & 139 \\
    Sao Tome and Principe & 138   & 138   & 143   & 143   & 143   & 142 \\
    Haiti & 135   & 139   & 136   & 136   & 136   & 137 \\
    Ghana & 137   & 140   & 140   & 140   & 140   & 140 \\
    Bermuda & 126   & 141   & 144   & 144   & 144   & 144 \\
    Sierra Leone & 148   & 142   & 145   & 145   & 145   & 145 \\
    Papua New Guinea & 117   & 143   & 142   & 141   & 141   & 143 \\
    San Marino & 119   & 144   & 137   & 138   & 138   & 135 \\
    Burundi & 143   & 145   & 146   & 146   & 146   & 146 \\
    Rwanda & 142   & 146   & 148   & 148   & 148   & 148 \\
    U.S. Virgin Islands & 149   & 147   & 149   & 149   & 149   & 149 \\
    Seychelles & 136   & 148   & 147   & 147   & 147   & 147 \\
    Senegal & 147   & 149   & 122   & 125   & 125   & 121 \\
\end{xtabular}
\end{center}


\end{document}